\documentclass[a4paper,11pt]{article}

 \makeatletter
\def\@fnsymbol#1{\ensuremath{\ifcase#1\or *\or
   \mathsection\or \mathparagraph\or \|\or **\or \dagger\dagger
   \or \ddagger\ddagger \else\@ctrerr\fi}}
    \makeatother

\usepackage{graphicx}

\usepackage{fullpage}
\usepackage{libertine}
\usepackage{color}

\usepackage[ruled,vlined]{algorithm2e}
\SetArgSty{textrm}

\usepackage{amsmath,amsfonts,amssymb,amsthm}

\usepackage[breaklinks=true]{hyperref}
\usepackage[svgnames]{xcolor}
\usepackage[capitalise,nameinlink]{cleveref}
\hypersetup{colorlinks={true},linkcolor={DarkBlue},citecolor=[named]{DarkGreen}}

\usepackage{natbib}

\usepackage{subcaption}

\usepackage{tikz}  
\usetikzlibrary{arrows}
\usetikzlibrary{patterns,snakes}
\usetikzlibrary{decorations.shapes}
\tikzstyle{overbrace text style}=[font=\tiny, above, pos=.5, yshift=5pt]
\tikzstyle{overbrace style}=[decorate,decoration={brace,raise=5pt,amplitude=3pt}]
\usetikzlibrary{shapes.geometric}

\newtheorem{theorem}{Theorem}[section]

\newtheorem{lemma}[theorem]{Lemma}

\theoremstyle{definition}
\newtheorem*{comment*}{Comment}

\newcommand{\cost}{\text{cost}}

\newcommand{\SC}{\text{SC}}

\newcommand{\bw}{\mathbf{w}}
\newcommand{\bo}{\mathbf{o}}

\newcommand{\bx}{\mathbf{x}}

\title{\bf Agent-Constrained Truthful Facility Location Games\thanks{A preliminary version of this paper will appear in the Proceedings of the 17th International Symposium on Algorithmic Game Theory (SAGT), 2024.}}

\usepackage{authblk}
\author[1]{Argyrios Deligkas}
\author[2]{Mohammad Lotfi}
\author[3]{Alexandros A. Voudouris\thanks{Corresponding author. Email: alexandros.voudouris@essex.ac.uk}}

\affil[1]{Royal Holloway University of London, UK}
\affil[2]{Sharif University of Technology, Iran}
\affil[3]{University of Essex, UK}

\date{}

\begin{document}

\allowdisplaybreaks

\maketitle

\begin{abstract}
We consider a truthful facility location game in which there is a set of agents with private locations on the line of real numbers, and the goal is to place a number of facilities at different locations chosen from the set of those reported by the agents. Given a feasible solution, each agent suffers an individual cost that is either its total distance to all facilities (sum-variant) or its distance to the farthest facility (max-variant). For both variants, we show tight bounds on the approximation ratio of strategyproof mechanisms in terms of the social cost, the total individual cost of the agents. 

\medskip
\noindent 
{\bf Keywords:} Mechanism design; Facility location; Approximation ratio.
\end{abstract}

\section{Introduction}\label{sec:intro}

Suppose you are the mayor of a small town and your task is to decide where to build a park and a library on a very busy street to accommodate the needs of the citizens. One way to make this decision is to simply place the facilities arbitrarily. Even though this is easy to implement, the chosen locations might not be very accessible and the citizens most probably will end up complaining and not vote for you in the next election. Instead, you could ask the citizens to suggest the possible locations where the facilities could be built and choose one that collectively satisfies them. While this now seems sufficient enough to get you re-elected, you also need to make sure that the citizens are incentivized to truthfully suggest their real ideal locations and not lie in order to minimize the distance they have to walk. This is known as the {\em truthful facility location game}. 

Since the seminal work of \citet{procaccia09approximate} on {\em approximate mechanism design without money}, many different variants of the game have been studied under assumptions about the number of facilities to be placed, the preferences of the agents for the facilities, and the feasible locations where the facilities can be built; we refer the reader to the survey of \citet{fl-survey} for an overview, and to our discussion of related work below. In this work we consider a previously unexplored, yet fundamental model where the facilities can be built at locations that are dynamically proposed by the agents, in contrast to previously studied models where the facilities could be placed either at any location on the line or only at a predetermined set of fixed candidate locations. 

\subsection{Our Model}
We consider the following agent-constrained truthful facility location game. An instance $I$ consists of a set of $n \geq 2$ \textit{agents} with \textit{private locations} on the line of real numbers, and $k \geq 2$ \textit{facilities} that can be placed at {\em different} locations chosen from the (multi-)set of locations reported by the agents. 
Given a feasible {\em solution} $\bx$ which determines the agent locations where the $k$ facilities are placed, each agent $i$ suffers an \textit{individual cost}. We consider two different models that differ on the cost function of the agents. 
In the {\em sum-variant}, the cost of $i$ in instance $I$ is its {\em total} distance from the facilities: 
$$\cost_i^{\text{sum}}(\bx|I) = \sum_{x \in \bx} d(i,x),$$
where $d(i,x) = |i-x|$ is the distance between the location of agent $i$ and point $x$ on the line. 
In the {\em max-variant}, the cost of $i$ in instance $I$ is its distance to the {\em farthest} facility: 
$$\cost_i^{\text{max}}(\bx|I) = \max_{x \in \bx}\{d(i,x)\}.$$
Whenever the variant we study is clear from context, we will drop the sum and max from notation, and simply write $\cost_i(\bx|I)$ for the individual cost of $i$ when solution $\bx$ is chosen; similarly, we will drop $I$ from notation when the instance is clear from context and write $\cost_i(\bx)$.
We are interested in choosing solutions that have a small effect in the overall cost of the agents, which is captured by the {\em social cost} objective function, defined as:
$$\SC(\bx|I) = \sum_i \cost_i(\bx|I).$$
A solution can also be {\em randomized} in the sense that it is a probability distribution $\mathbf{p} = (p_\bx)_{\bx}$ over all feasible solutions; the {\em expected social cost} of such a randomized solution is defined appropriately as 
$$\mathbb{E}[\SC(\mathbf{p}|I)] = \sum_{\bx} p_\bx \cdot \SC(\bx|I).$$

The solution is decided by a {\em mechanism} based on the locations reported by the agents; let $M(I)$
be the solution computed by a mechanism $M$ when given as input an instance $I$. A mechanism $M$ is said to be {\em strategyproof} if no agent $i$ can misreport its true location and decrease its individual cost; that is, 
$$\cost_i(M(I)|I) \leq \cost_i(M(J)|I)$$ 
for every pair of instances $I$ and $J$ that differ only on the location reported by agent $i$. In case the mechanism is randomized, then it is said to be {\em strategyproof-in-expectation} if no agent $i$ cannot misreport its true location and decrease its {\em expected} individual cost.

The {\em approximation ratio} of a mechanism is the worst-case ratio (over all possible instances) of the (expected) social cost of the chosen solution over the minimum possible social cost: 
\begin{align*}
    \sup_{I} \frac{\mathbb{E}[\SC(M(I)|I)]}{\min_{\bx} \SC(\bx|I)}.
\end{align*}
Our goal is to design mechanisms that are strategyproof and achieve an as small approximation ratio as possible. 


\subsection{Our Contribution}
For both individual cost variants, we show tight bounds on the best possible approximation ratio that can be achieved by strategyproof mechanisms. We start with the case of $k=2$ facilities for which we study both deterministic and randomized mechanisms. For the sum-variant, in Section~\ref{sec:sum}, we show a tight bound of $3/2$ for deterministic mechanisms and a bound of $10-4\sqrt{5}\approx 1.0557$ for randomized ones. For the max-variant, in Section~\ref{sec:max}, we show bounds of $3$ and $2$ on the approximation ratio of deterministic and randomized mechanisms, respectively. An overview of our bounds for $k=2$ is given in Table~\ref{tab:overview}.
In Section~\ref{sec:multiple}, we switch to the general case of $k$ facilities and focus exclusively on deterministic mechanisms. For the sum-variant, we show that the approximation ratio is between $2-1/k$ and $2$, while for the max-variant, we show a tight bound of $k+1$. 

\renewcommand{\arraystretch}{1.3}
\begin{table}[t]
    \centering
    \begin{tabular}{c|cc}
                      & Deterministic & Randomized \\ 
       \hline
       sum-variant    & $3/2$ & $1.0557$ \\
       max-variant    & $3$   & $2$ \\
       \hline
    \end{tabular}
    \caption{An overview of the tight bounds on the approximation ratio of deterministic and randomized strategyproof mechanisms for the sum- and max-variants when $k=2$.}
    \label{tab:overview}
\end{table}

Our upper bounds follow by appropriately defined {\em statistic-type} mechanisms that choose the agent locations where the facilities will be placed according to the ordering of the agents on the line from left to right. In particular, for $k=2$, our mechanisms locate one facility at the median agent $m$ and the other either at the agent $\ell$ that is directly to the left of $m$ or the agent $r$ that is directly to the right of $m$. To be even more specific, our deterministic mechanism always chooses the solution $(m,r)$, while our randomized mechanisms choose the solutions $(\ell,m)$ and $(m,r)$ according to some probability distribution. Interestingly, for the sum-variant, it turns out that the probabilities are functions of the distances $d(\ell,m)$ and $d(m,r)$; to the best of our knowledge, this is one of few settings in which the best possible randomized strategyproof mechanism is not required to assign fixed, constant probabilities. For the general case of $k$ facilities, our (deterministic) upper bounds for both variants follow by a mechanism that is a natural generalization of the one for $k=2$; in particular, the mechanism places the facilities around the median agent(s) within a radius of about $k/2$. 

\subsection{Related Work}
Truthful facility location games have a long history within the literature of {\em approximate mechanism design without money}, starting with the paper of \citet{procaccia09approximate}. Various different models have been studied depending on parameters such as
the number of facilities whose location needs to be determined~\citep{procaccia09approximate,Lu2010two-facility,fotakis2014two}, 
whether the facilities are obnoxious~\citep{cheng2013obnoxious}, 
whether the agents have different types of preferences over the facilities (for example, optional~\citep{chen2020max,kanellopoulos2023discrete,li2020constant,serafino2016}, fractional~\citep{fong2018fractional}, or hybrid~\citep{feigenbaum2015hybrid}), 
and whether there are other limitations or features (for example, the facilities might only be possible to be built at specific fixed locations~\citep{feldman2016voting,gai2024mixed,kanellopoulos2025,Xu2021minimum}, there might be limited resources that can be used to build some of the available facilities rather than all~\citep{deligkas2023limited}, there might be limited available information during the decision process~\citep{chan2023ordinal,filos2024distributed}, or there might be even more information in the form of predictions about the optimal facility locations which can be leveraged~\citep{agrawal2022predictions,Xu2022predictions}). 
We refer the reader to the survey of \citet{fl-survey} for more details on the different dimensions along which truthful facility location games have been studied over the years.

When there are multiple facilities to locate, the typical assumption about the individual behavior of the agents is that they aim to minimize their distance to the closest facility~\citep{procaccia09approximate,Lu2010two-facility,fotakis2014two,Tang2020candidate,Xu2022predictions}; such a cost model essentially assumes that the facilities are homogeneous (in the sense that they offer the same service) and thus each agent is satisfied if it is close enough to one of them. In contrast, both variants (sum and max) we consider here model different cases in which the facilities are heterogeneous (in the sense that they offer different services) and each agent aims to minimize either the total or the maximum distance to the facilities. These variants have also been considered in previous works under different assumptions than us; in particular, the sum-variant has been studied by \citet{serafino2016,kanellopoulos2025,gai2024mixed} and \citet{Xu2021minimum}, while the max-variant has been studied by \citet{chen2020max,Zhao2023constrained} and \citet{lotfi2024max}. 

The main differences between our work and the aforementioned ones are the following: 
In most of these papers, the agents have optional preferences over the facilities; that is, some agents approve one facility and are indifferent to the other, while some agents approve both facilities. Here, we focus exclusively on the fundamental case where all agents approve both facilities. In addition, some of these papers study a constrained model according to which the facilities can only be built at different locations chosen from a set of fixed, predetermined candidate ones. In our model, the facilities can also only be built at different locations, which, however, are chosen from the set of locations that are reported by the agents; this is a more dynamic setting in the sense that the candidate locations could change if agents misreport.
\footnote{We highlight that that the two models are different in terms of the results as well. For example, in the sum-variant, we here have shown a tight bound of $3/2$ for the case of two facilities, while previous work has shown a lower bound of $1+\sqrt{2}\approx 2.42$ on the approximation ratio of deterministic strategyproof mechanisms for the candidate location model~\citep{kanellopoulos2025}.} 
We remark that, in continuous facility location settings (where the facilities can be placed anywhere on the line) such as those studied in the original paper of \citet{procaccia09approximate} and follow-up work, the class of strategyproof mechanisms consists of mechanisms that place the facilities at agent locations (according to their ordering) and mechanisms that place the facilities at fixed locations disregarding the agent input. However, to the best of our knowledge, there has not been any previous work that has studied the model where the candidate locations are restricted to the ones reported by the agents and feasible solutions are those where the facilities are placed at different agent locations, an assumption that crucially affects the social cost value of the optimal solution.

\section{Sum-variant for Two Facilities} \label{sec:sum}
We start the presentation of our technical results with the case of $k=2$ facilities and the sum-variant. Recall that in this variant the individual cost of any agent is its distance from both facilities. We will first argue about the structure of the optimal solution; this will be extremely helpful in bounding the approximation ratio of our strategyproof mechanisms later on. We start with the case where the number of agents $n$ is an even number, for which the optimal solution is well-defined and actually leads to an optimal strategyproof mechanism. 

\begin{lemma}\label{lem:sum:even}
For any even $n \geq 2$, an optimal solution is to place the facilities at the two median agents.
\end{lemma}

\begin{proof}
Let $m_1$ and $m_2$ be two median agents. Suppose that there is an optimal solution $(o_1,o_2)$ with $o_1 \leq o_2$. Since any point $x \in [m_1, m_2]$ minimizes the total distance of all agents from any other point of the line, we have
\begin{align*}
\SC(m_1,m_2) = \sum_i d(i,m_1) + \sum_i d(i,m_2) \leq \sum_i d(i,o_1) + \sum_i d(i,o_2) = \SC(o_1,o_2),
\end{align*}
and thus $(m_1,m_2)$ is also an optimal solution. 
\end{proof}

Before we continue, we remark that the {\sc Two-Medians} mechanism, which is implied by Lemma~\ref{lem:sum:even}, is indeed strategyproof: 
To change the solution of the mechanism, an agent $i$ would have to report a location $x > m_1$ in case $i\leq m_1$ or a location $x<m_2$ in case $i \geq m_2$; such a misreport leads to an individual cost of at least $\min\{d(i,x),d(i,m_2)\} + d(i,m_2)$ in the first case and of at least $d(i,m_1)+\min\{d(i,x),d(i,m_1)\}$ in the second case, which is at least the true individual cost $d(i,m_1) +d(i,m_2)$ of $i$. Hence, agent $i$ has no incentive to deviate and the mechanism is strategyproof.

For the case where the number of agents $n \geq 3$ is an odd number, it will be useful to calculate the social cost of the solutions $(\ell,m)$ and $(m,r)$, where $\ell$ and $r$ are the agents directly to the left and right of the median agent $m$, respectively. 

\begin{lemma}\label{lem:social-cost-x}
For any  $x \in \{\ell,r\}$, the social cost of the solution $(x,m)$ is 
$$\SC(x,m) = 2 \cdot \sum_i d(i,m) + d(m,x).$$
\end{lemma}

\begin{proof}
By the definition of $m$ and since there is an odd number of agents, we have $|\{i \leq \ell\}| = |\{i \geq r\}|$.
Thus, we can match each agent $i \leq \ell$ to a unique agent $\mu(i) \geq r$. Since $i \leq \ell \leq m \leq r \leq \mu(i)$, we have that 
$$d(i,\ell) + d(\mu(i),\ell) = d(i,r) + d(\mu(i),r) = d(i,m) + d(\mu(i),m).$$
Hence, for any $x \in \{\ell,r\}$, 
\begin{align*}
\SC(x,m) 
&= \sum_i \bigg( d(i,m) + d(i,x) \bigg) \nonumber \\
&= \sum_i d(i,m) + \sum_{i \leq \ell} \bigg(d(i,x) + d(\mu(i),x) \bigg) + d(m,x) \nonumber \\
&= \sum_i d(i,m) + \sum_{i \leq \ell} \bigg(d(i,m) + d(\mu(i),m) \bigg) + d(m,x) \nonumber \\
&= 2 \cdot \sum_i d(i,m) + d(m,x).
\end{align*}
This completes the proof. 
\end{proof}

Using the above lemma, we can argue about the structure of the optimal solution. 

\begin{lemma}\label{thm:sum:odd}
For any odd $n \geq 3$, an optimal solution is to place the facilities at the median agent and the agent that is closest to it.
\end{lemma}

\begin{proof}
Clearly, one of $\ell$ or $r$ is the closest agent to $m$, say $\ell$; hence, $d(\ell,m) \leq d(m,r)$. 
To simplify our notation, for any $x$ let $f(x) = \sum_i d(i,x)$ denote the total distance of all agents from $x$. It is well-known that $f$ is monotone such that $f(i) \geq f(\ell) \geq f(m)$ for every $i \leq \ell \leq m$, and $f(i) \geq f(r) \geq f(m)$ for every $i \geq r \geq m$. Consequently, the optimal solution is either $(\ell,m)$ or $(m,r)$. By Lemma~\ref{lem:social-cost-x} with $x=\ell$ and $x=r$, we get
$$\SC(\ell,m) - \SC(m,r) = d(\ell,m) - d(m,r).$$ 
Since $d(\ell,m) \leq d(m,r)$, we conclude that $\SC(\ell,m) \leq \SC(m,r)$ and the solution $(\ell,m)$ is indeed the optimal one. 
\end{proof}

It is not hard to observe that when $n$ is odd, computing the optimal solution is not strategyproof; the second-closest agent to the median might have incentive to misreport a location slightly closer to the median to move the second facility there. 
However, we do know that one of the solutions $(\ell,m)$ and $(m,r)$ must be optimal. Based on this, we consider the following {\sc Median-Right} mechanism: Place one facility at the position the median agent $m$ and the other at the position of the agent $r$ directly to the right of $m$.\footnote{Clearly, since we are dealing with the case of odd $n$, instead of this mechanism, one could also consider the {\sc Median-Left} mechanism which places the second facility to the agent $\ell$ that is directly to the left of $m$; both mechanisms are symmetric and achieve the same approximation ratio.} One can verify that this mechanism is strategyproof using an argument similar to the one we presented above for the {\sc Two-Medians} mechanism in the case of even $n$. So, we continue by bounding its approximation ratio. 

\begin{theorem}\label{thm:sum:deterministic:upper}
For any odd $n \geq 3$, the approximation ratio of the {\sc Median-Right} mechanism is at most $3/2$.
\end{theorem}

\begin{proof}
The solution of the mechanism is $\bw = (m,r)$. If $r$ is the closest agent to $m$, then $\bw$ is optimal by Lemma~\ref{thm:sum:odd}. So, assume that this is not the case and the optimal solution is $\bo = (\ell,m)$. By Lemma~\ref{lem:social-cost-x} with $x=r$, we get
\begin{align*}
\SC(\bw) = 2 \cdot \sum_i d(i,m) + d(m,r).
\end{align*}
Similarly, for $x=\ell$, we get
\begin{align*}
\SC(\bo) &= 2 \cdot \sum_i d(i,m) + d(\ell,m) \\
&\geq 2 \cdot \sum_i d(i,m) \\
&\geq 2\cdot |\{i \geq r\}| \cdot d(m,r) = (n-1) \cdot d(m,r).  
\end{align*}
Using these two lower bounds on the optimal social cost, we can now upper-bound the social cost of $\bw$ as follows:
\begin{align*}
\SC(\bw) \leq \bigg( 1 + \frac{1}{n-1} \bigg) \cdot \SC(\bo) = \frac{n}{n-1} \cdot \SC(\bo).
\end{align*}
Therefore, the approximation ratio is at most $n/(n-1) \leq 3/2$ for any $n \geq 3$.
\end{proof}

The approximation ratio of $3/2$ is in fact the best possible that can be achieved by any deterministic strategyproof mechanism.

\begin{theorem}\label{thm:sum:deterministic:lower}
The approximation ratio of any deterministic strategyproof mechanism is at least $3/2$.
\end{theorem}

\begin{proof}
Consider any deterministic strategyproof mechanism and an instance with $n=3$ agents located at $0$, $1$ and $2$. 
Since there are three possible locations for two facilities, one facility must be placed at $0$ or $2$, say $0$. 
Then, the cost of the agent $i$ that is located at $2$ is at least $2$ (in particular, the cost of this agent is $3$ if the solution is $(0,1)$ and $2$ if the solution is $(0,2)$).

Now suppose that $i$ moves to $1+\varepsilon$ for some infinitesimal $\varepsilon > 0$. Due to strategyproofness, the mechanism must place one of the facilities at $0$ in the new instance as well. Otherwise, agent $i$ would have cost $2-\varepsilon$ and would prefer to misreport its position as $1+\varepsilon$ instead of $2$. So, the social cost of any of the two possible solutions (either $(0,1)$ or $(0,1+\varepsilon)$) that the mechanism can output is approximately $3$. In contrast, the optimal solution is $(1,1+\varepsilon)$ with social cost approximately $2$, leading to an approximation ratio of $3/2$.
\end{proof}

Since the optimal solution is either $(\ell,m)$ or $(m,r)$, it is reasonable to think that randomizing over these two solutions, rather than blindly choosing one of them, can lead to an improved approximation ratio. Indeed, we can show a significantly smaller tight bound of $10-4\sqrt{5} \approx 1.0557$ for randomized strategyproof mechanisms when $n \geq 3$ is an odd number; recall that, for even $n \geq 2$, we can always compute the optimal solution. 
For the upper bound, we consider the following {\sc Reverse-Proportional} randomized mechanism: 
With probability $p_\ell = \frac{d(m,r)}{d(\ell,r)}$ choose the solution $(\ell,m)$, and
with probability $p_r = \frac{d(\ell,m)}{d(\ell,r)}$ choose the solution $(m,r)$.

\begin{theorem} \label{thm:sum:randomized-sp}
The {\sc Reverse-Proportional} mechanism is strategyproof-in-expectation.
\end{theorem}

\begin{proof}
Due to symmetry, it suffices to show that no agent $i \geq m$ has any profitable deviation. We first consider agent $m$, and then any agent $i \geq r$. 

\medskip

\noindent 
{\bf Agent $m$.} The expected individual cost of $m$ is 
\begin{align*}
\cost_m(\ell,m,r) = \frac{d(m,r)}{d(\ell,r)} \cdot d(\ell,m) + \frac{d(\ell,m)}{d(\ell,r)} \cdot d(m,r) 
= 2\cdot \frac{d(\ell,m)\cdot d(m,r)}{d(\ell,r)}.
\end{align*}
Suppose now that $m$ deviates to another location $x$ such that $\ell \leq x < m \leq r$. In this new instance, the solution $(\ell,x)$ is chosen with probability $d(x,r)/d(\ell,r)$ and the solution $(x,r)$ is chosen with probability $d(\ell,x)/d(\ell,r)$. The expected cost of $m$ becomes
\begin{align*}
\cost_m(\ell,x,r) 
&= \frac{d(x,r)}{d(\ell,r)} \cdot \bigg( d(\ell,m) + d(x,m) \bigg) + \frac{d(\ell,x)}{d(\ell,r)} \cdot \bigg( d(x,m) + d(m,r) \bigg) \\
&= \frac{d(x,r)}{d(\ell,r)} \cdot d(\ell,m) + \frac{d(\ell,x)}{d(\ell,r)} \cdot d(m,r) + d(x,m) \\
&= \frac{d(x,m) + d(m,r)}{d(\ell,r)} \cdot d(\ell,m) + \frac{d(\ell,m)-d(x,m)}{d(\ell,r)} \cdot d(m,r) + d(x,m) \\
&= \cost_m(\ell,m,r) + \frac{d(x,m)}{d(\ell,r)} \cdot d(\ell,m) - \frac{d(x,m)}{d(\ell,r)} \cdot d(m,r) + d(x,m) \\
&= \cost_m(\ell,m,r) + \frac{d(x,m)}{d(\ell,r)} \cdot \bigg( d(\ell,r) + d(\ell,m) - d(m,r) \bigg) \\
&= \cost_m(\ell,m,r) + \frac{d(x,m)}{d(\ell,r)} \cdot 2 d(\ell,m), 
\end{align*}
and thus $m$ has no incentive to deviate to such a location $x$. 

Next, suppose that $m$ deviates to a location $x$ such that $x < \ell \leq m \leq r$. In this new instance, the solution $(x,\ell)$ is chosen with probability $d(\ell,r)/d(x,r)$ and the solution $(\ell,r)$ is chosen with probability $d(x,\ell)/d(x,r)$. The expected cost of $m$ becomes
\begin{align*}
\cost_m(x,\ell,r) 
&= \frac{d(\ell,r)}{d(x,r)} \cdot \bigg( d(x,m) + d(\ell,m) \bigg) + \frac{d(x,\ell)}{d(x,r)} \cdot \bigg( d(\ell,m) + d(m,r) \bigg) \\
&= d(\ell,m) + \frac{d(\ell,r)}{d(x,r)} \cdot d(x,m) +  \frac{d(x,\ell)}{d(x,r)} \cdot d(m,r) \\
&= d(\ell,m) + \frac{d(\ell,r)}{d(x,\ell) + d(\ell,r)} \cdot \bigg( d(x,\ell) + d(\ell,m) \bigg) +  \frac{d(x,\ell)}{d(x,r)} \cdot d(m,r). 
\end{align*}
As a function of $d(\ell,x) > 0$, $\cost_m(x,\ell,r)$ is a non-decreasing function and thus
\begin{align*}
\cost_m(x,\ell,r) > 2\cdot d(\ell,m) \geq 2\cdot \frac{d(\ell,m)\cdot d(m,r)}{d(\ell,r)} = \cost_m(\ell,m,r),
\end{align*}
where the second inequality follows since $d(m,r)/d(\ell,r) \leq 1$. So, $m$ again has no incentive to deviate to such a location $x$.

\medskip
\noindent 
{\bf Agent $i \geq r$.} 
The expected individual cost of $i$ is
\begin{align*}
\cost_i(\ell,m,r) 
&= \frac{d(m,r)}{d(\ell,r)} \cdot \bigg( d(\ell,i) + d(m,i) \bigg) + \frac{d(\ell,m)}{d(\ell,r)} \cdot \bigg( d(m,i) + d(r,i) \bigg)  \\
&= d(m,i) + \frac{d(m,r)}{d(\ell,r)} \cdot \bigg( d(\ell,r) + d(r,i) \bigg) + \frac{d(\ell,m)}{d(\ell,r)} \cdot d(r,i) \\
&= d(m,i) + d(m,r) + d(r,i) \\
&= 2 \cdot d(m,i).
\end{align*}
First consider a deviation of $i$ to a location that retains the order of $\ell$ and $m$ but changes the location of the agent directly to the right of $m$ to $x \geq m$. That is, the three middle agents have locations $\ell$, $m$, and $x$.
In this new instance, the solution $(\ell,m)$ is chosen with probability $d(m,x)/d(\ell,x)$ and the solution $(m,x)$ is chosen with probability $d(\ell,m)/d(\ell,x)$. Hence, the expected individual cost of $i$ is 
\begin{align*}
\cost_i(\ell,m,x) 
&= \frac{d(m,x)}{d(\ell,x)} \cdot \bigg( d(\ell,i) + d(m,i) \bigg) + \frac{d(\ell,m)}{d(\ell,x)} \bigg( d(m,i) + d(x,i) \bigg)  \\
&= d(m,i) + \frac{d(m,x)}{d(\ell,x)} \cdot d(\ell,i) + \frac{d(\ell,m)}{d(\ell,x)} \cdot d(x,i).
\end{align*}
We now consider the following two cases depending on the relative positions of $r$ and $x$. 
\begin{itemize}
\item If $m \leq x \leq r \leq i$, then since $d(\ell,i) = d(\ell,x) + d(x,i)$ and $d(m,i) = d(m,x) + d(x,i)$, we have 
\begin{align*}
\cost_r(\ell,m,x) 
&= d(m,i) + \frac{d(m,x)}{d(\ell,x)} \cdot \bigg( d(\ell,x) + d(x,i) \bigg)  + \frac{d(\ell,m)}{d(\ell,x)} \cdot d(x,i) \\
&= d(m,i) + d(m,x) + d(x,i)  \\
&= 2 \cdot d(m,i), 
\end{align*}
and thus $i$ does not decrease its cost.

\item If $m \leq r < x$, then it must be the case that $i=r$ since no agent $i > r$ can deviate to location $x > r$ and be closer to $m$ than $r$. Since $d(x,r) > 0$, $d(m,x) = d(m,r) + d(x,r)$, $d(\ell,r) \geq d(m,r)$ and $d(\ell,r) + d(x,r) = d(\ell,x)$, we have
\begin{align*}
\cost_r(\ell,m,x) 
&> d(m,r) + \frac{d(m,r) + d(x,r)}{d(\ell,x)} \cdot d(\ell,r) \\
&= d(m,r) + \frac{d(m,r)}{d(\ell,x)} \cdot d(\ell,r)  + \frac{d(\ell,r)}{d(\ell,x)} \cdot d(x,r) \\
&\geq d(m,r) + \frac{d(m,r)}{d(\ell,x)} \cdot d(\ell,r)  + \frac{d(m,r)}{d(\ell,x)} \cdot d(x,r) \\
&= d(m,r) + \frac{d(m,r)}{d(\ell,x)} \cdot \bigg( d(\ell,r) + d(x,r) \bigg) \\
&= 2 \cdot d (m,r). 
\end{align*}
Hence, again $r$ does not decrease its cost. 
\end{itemize}

Finally, consider the case where agent $i \geq r$ deviates to a location at the left of $m$ such that the three middle agents have locations 
$x \leq \ell$, $y \leq m$ and $z \leq m$ instead of $\ell$, $m$ and $r$. In this new instance, the solutions $(x,y)$ and $(y,z)$ are chosen with positive probability. Since $x,y,z \leq m$, the individual cost of $i$ for the solution $(x,y)$ is $d(x,i) + d(y,i) \geq 2\cdot d(m,i)$ and, similarly, the individual cost of $i$ for the solution $(y,z)$ is $d(y,i) + d(z,i) \geq 2\cdot d(m,i)$. Consequently, the expected cost of $i$ is at least $2 \cdot d(m,i) = d_i(\ell,m,r)$ for any probability distribution over the solutions $(x,y)$ and $(y,z)$, which means that $i$ has no incentive to deviate. 
\end{proof}

\begin{theorem} \label{thm:sum:randomized-upper}
For any odd $n \geq 3$, the approximation ratio of the {\sc Reverse-Proportional} mechanism is at most $10-4\sqrt{5} \approx 1.0557$.
\end{theorem}

\begin{proof}
Without loss of generality, suppose that $d(\ell,m) \leq d(m,r)$ and thus the optimal solution is $\bo = (\ell,m)$. 
By the definition of the mechanism, the solutions $d(\ell,m)$ and $d(m,r)$ are chosen with probability $p_\ell = d(m,r)/d(\ell,r)$ and $p_r = d(\ell,m)/d(\ell,r)$, respectively; observe that $p_\ell \geq p_r$. By Lemma~\ref{lem:social-cost-x} with $x=\ell$ and using the fact that that $d(\ell,m) + d(m,r) = d(\ell,r)$, we can lower-bound the optimal social cost as follows:
\begin{align*}
\SC(\bo) &= 2 \cdot \sum_i d(i,m) + d(\ell,m) \geq 2\cdot d(\ell,r) + d(\ell,m).
\end{align*}
Again using Lemma~\ref{lem:social-cost-x} with $x = \ell$ and $x=r$, as well as the fact that $p_\ell = 1-p_r$, we can write the expected social cost of the randomized solution $\bw$ chosen by the mechanism as
\begin{align*}
\mathbb{E}[\SC(\bw)]
&= p_\ell \cdot \bigg( 2 \cdot \sum_i d(i,m) + d(\ell,m) \bigg) + p_r \cdot \bigg( 2 \cdot \sum_i d(i,m) + d(m,r) \bigg) \\
&= 2 \cdot \sum_i d(i,m) + (1-p_r) \cdot d(\ell,m) + p_r \cdot d(m,r) \\
&= 2 \cdot \sum_i d(i,m) + d(\ell,m) + p_r \cdot \bigg( d(m,r) - d(\ell,m) \bigg) \\
&= \SC(\bo) + p_r \cdot \bigg( d(m,r) - d(\ell,m) \bigg).
\end{align*}
Consequently, the approximation ratio is
\begin{align*}
\frac{\mathbb{E}[\SC(\bw)]}{\SC(\bo)} &\leq 1 + p_r \cdot \frac{d(m,r) - d(\ell,m)}{2\cdot d(\ell,r) + d(\ell,m)} \\
&= 1 + p_r \cdot \frac{\frac{d(m,r)}{d(\ell,r)} - \frac{d(\ell,m)}{d(\ell,r)}}{2 + \frac{d(\ell,m)}{d(\ell,r)}} \\
&= 1 + p_r \cdot \frac{p_\ell - p_r}{2 + p_r}
\end{align*}
Using the fact that $p_\ell = 1 - p_r$, we finally have that
\begin{align*}
\frac{\mathbb{E}[\SC(\bw)]}{\SC(\bo)} &\leq 1 + p_r \cdot \frac{1-2\cdot p_r}{2 + p_r}.
\end{align*}
The last expression attains its maximum value of $10-4\sqrt{5} \approx 1.0557$ for $p_r = \sqrt{5}-2$.
\end{proof}

Next, we will argue that the {\sc Reverse-Proportional} mechanism is the best possible by showing a matching lower bound on the approximation ratio of any randomized strategyproof-in-expectation mechanism. To do this, we will use instances with three agents for which we first show the following technical lemma that reduces the class of mechanisms to consider. 

\begin{lemma} \label{lem:sum:randomized-structure}
Consider any instance with three agents located at $x < y < z$. For any randomized mechanism $M$ that assigns positive probability to the solution $(x,z)$, there exists a randomized mechanism $M_0$ that assigns $0$ probability to that solution and achieves at most as much expected social cost as $M$.
\end{lemma}

\begin{proof}
Let $p_{xy}$, $p_{yz}$ and $p_{xz}$ be the probabilities assigned to the three possible solutions $(x,y)$, $(y,z)$ and $(x,z)$, with $p_{xz} > 0$. The social cost of each solution is 
\begin{align*}
\SC(x,y) &= d(x,y) + d(x,y) + d(x,z) + d(y,z) = 3 \cdot d(x,y) + 2\cdot d(y,z), \\
\SC(y,z) &= d(x,y) + d(x,z) + d(y,z) + d(y,z) = 2 \cdot d(x,y) + 3\cdot d(y,z), \\
\SC(x,z) &= d(x,z) + d(x,y) + d(y,z) + d(x,z) = 3 \cdot d(x,y) + 3\cdot d(y,z). 
\end{align*}
So, the expected social cost of the randomized solution $\mathbf{p} = (p_{xy}, p_{yz}, p_{xz})$ computed is
\begin{align*}
\mathbb{E}[\SC(\mathbf{p})] &= \bigg(3p_{xy} + 2p_{yz} + 3p_{xz} \bigg) \cdot d(x,y) + \bigg( 2p_{xy} + 3p_{yz} + 3p_{xz} \bigg) \cdot d(y,z) \\
&= 2\cdot \bigg( d(x,y) + d(y,z) \bigg) + 
\bigg(p_{xy} + p_{xz} \bigg) \cdot d(x,y) + \bigg( p_{yz} + p_{xz} \bigg) \cdot d(y,z).
\end{align*}
Clearly, for any $\lambda \in (0,1)$, 
\begin{align*}
\mathbb{E}[\SC(\mathbf{p})] 
&\geq 2\cdot \bigg( d(x,y) + d(y,z) \bigg) + 
\bigg(p_{xy} + \lambda \cdot p_{xz} \bigg) \cdot d(x,y) + \bigg( p_{yz} + (1-\lambda) p_{xz} \bigg) \cdot d(y,z).
\end{align*}
The last expression is exactly equal to the expected social cost of the randomized solution $\mathbf{q}$ that assigns probabilities $q_{xy} = p_{xy} + \lambda \cdot p_{xz}$, $q_{yz} = p_{yz} + (1-\lambda) p_{xz}$, and $q_{xz} = 0$, thus showing the claim that $\mathbb{E}[\SC(\mathbf{p})] \geq \mathbb{E}[\SC(\mathbf{q})]$.
\end{proof}

Using the above lemma, we can now show the desired lower bound.

\begin{theorem}\label{thm:sum:randomized-lower}
For the sum-variant, the approximation ratio of any randomized strategyproof-in-expectation mechanism is at least $10-4\sqrt{5} \approx 1.0557$.
\end{theorem}

\begin{proof}
Consider any randomized strategyproof mechanism and an instance $I$ with three agents located at $0$, $1$ and $2$. Let $p_0(I)$ and $p_1(I)$ be the probabilities assigned to solutions $(0,1)$ and $(1,2)$, respectively. 
By Lemma~\ref{lem:sum:randomized-structure}, we can assume that $p_0(I) + p_1(I) = 1$, and thus suppose that $p_0(I) \geq 1/2$ without loss of generality. The expected individual cost of the agent $i$ that is located at $2$ is then 
$$3\cdot p_0(I) + 1\cdot p_1(I) \cdot 1 = 3\cdot p_0(I) + 1-p_0(I) = 2\cdot p_0(I) + 1 \geq 2.$$

Now consider an instance $J$ with three agents located at $0$, $1$ and $x=1/q \in (1,2)$, where $q=3-\sqrt{5} \approx 0.764$; hence, the only different between $I$ and $J$ is that agent $i$ is now located at $x$ rather than $2$. 
Let $p_0(J)$ and $p_x(J)$ be the probabilities assigned to solutions $(0,1)$ and $(1,x)$, respectively. 
Again, using Lemma~\ref{lem:sum:randomized-structure} we can assume that $p_0(J) + p_x(J) = 1$; any other case would achieve worse approximation ratio. 
Suppose that $p_x(J) > q$. 
Then, the expected cost of agent $i$ when misreporting its position as $1/q$ rather than $2$ would be 
\begin{align*}
3\cdot p_0(J) + \left(1 + 2-\frac{1}{q}\right)\cdot p_x(J)
&= 3 \cdot \bigg(1-p_x(J)\bigg) + \left(3 - \frac{1}{q} \right)\cdot p_x(J) \\
&= 3- \frac{1}{q} \cdot p_x(J) < 2
\end{align*}
and agent $i$ would manipulate the mechanism. Therefore, for the mechanism to be strategyproof, it has to be the case that $p_x(J) \leq q$, and thus $p_0(J) \geq 1-q$. 

In instance $J$, the optimal solution is $(1,x)$ with social cost $1+1/q + 2(1/q-1)=3/q-1$. Since the social cost of the solution $(0,1)$ is $2 + 1/q + 1/q-1 = 2/q+1$, the approximation ratio is 
\begin{align*}
\frac{p_0(J) \cdot \SC(0,1) + p_x(J)\cdot\SC(1,x)}{\SC(1,x)} 
&= p_x(J) + p_0(J) \cdot \frac{2/q+1}{3/q-1} \\
&= 1-p_0(J) + p_0(J) \cdot \frac{2+q}{3-q} \\
&= 1 +p_0(J) \cdot \frac{2q-1}{3-q} \\
&\geq 1 + (1-q) \cdot \frac{1-2(1-q)}{2 + (1-q)} = 10-4\sqrt{5}. 
\end{align*}
Hence, the approximation ratio is at least $10-4\sqrt{5} \approx 1.0557$.
\end{proof}

\section{Max-variant for Two Facilities} \label{sec:max}
We now turn our attention to the max-variant in which the individual cost of any agent is its distance from the farthest facility. One might be tempted to assume that the optimal solution has the same structure as in the sum-variant, which trivially holds for the case of $n=2$ agents. However, this is not true as the following example demonstrates: Consider an instance with $n=4$ agents with locations $-1/2$, $0$, $1$, and $2$. The optimal solution is $(-1/2,0)$ with a social cost of $5$; note that the two-medians solution $(0,1)$, which is optimal for the sum-variant according to Lemma~\ref{lem:sum:even}, has social cost $11/2$. 

In spite of this, we do not require the exact structure of the optimal solution to identify the best possible strategyproof mechanisms. For the class of deterministic mechanisms, we once again consider the {\sc Median-Right} mechanism; recall that this mechanism places one facility at the (leftmost) median agent $m$ and the other at agent $r$ that is directly to the right of $m$. This mechanism is strategyproof for the max-variant as well: The true individual cost of any agent $i \geq r$ is $d(i,m)$, and any misreport $x \geq m$ of does not change it, while any misreport $x < m$ can only lead to a larger cost; the case of $i < m$ is similar. We next show that this mechanism always achieves an approximation ratio of at most $3$, and it can achieve an improved approximation ratio of at most $2$ when the number of agents is even.\footnote{We remark that since {\sc Median-Right} is a strategyproof mechanism for both the sum- and the max-variant, the upper bound of $3$ on its approximation ratio for the max-variant follows directly from the upper bound of $3/2$ on its approximation ratio for the sum-variant; this is due to the sum and max individual cost functions being within a factor of $2$ of each other (see~\citep{lotfi2024max}). The bound can also be derived by setting $k=2$ to the more general bound of $k+1$ that we show for the case of multiple facilities in Section~\ref{sec:multiple:max}. Nevertheless, we include a full proof for completeness (and to also capture the case of even $n$ which will be useful later in the analysis of our randomized mechanism).}

\begin{theorem}\label{thm:max:deterministic:upper}
The approximation ratio of the {\sc Median-Right} mechanism is at most $2$ for any even $n \geq 4$ and at most $3$ for any odd $n \geq 3$. 
\end{theorem}

\begin{proof}
Let $\bw = (m,r)$ be the solution chosen by the mechanism, and denote by $\bo$ an optimal solution.
Observe that $\cost_i(\bw) = d(i,r)$ for every $i \leq m$ and $\cost_i(\bw) = d(i,m)$ for every $i \geq r$.
Hence,
\begin{align*}
\SC(\bw) 
&= \sum_{i\leq m} d(i,r) + \sum_{i \geq r} d(i,m) = \sum_{i} d(i,m) + |\{i \leq m\}| \cdot d(m,r).
\end{align*}
For the optimal solution $\bo$, since the location of the median agent is the point that minimizes the total distance from all agents, we have that
\begin{align*}
\SC(\bo) \geq \sum_{i} d(i,m).
\end{align*}
In addition, since there are two facilities to be placed, in $\bo$ one facility must be placed at the position of some agent $o \leq m$ or $o \geq r$. 
In the former case, we have that
$$\forall i \geq r: \cost_i(\bo) \geq d(i,o) = d(i,r) + d(m,r) + d(m,o) \geq d(m,r).$$
In the latter case, we have that
$$\forall i \leq m: \cost_i(\bo) \geq d(i,o) = d(i,m) + d(m,r) + d(r,o) \geq d(m,r).$$
Since $|\{i \leq m\}| \geq |\{i \geq r\}|$ by the definition of $m$ and $r$, we have established that, in any case,  
\begin{align*}
\SC(\bo) \geq |\{i \geq r\}|\cdot d(m,r).
\end{align*}
Using these two lower bounds on the optimal social cost, we can upper-bound the social cost of $\bw$ as follows:
\begin{align*}
\SC(\bw) \leq \bigg(1 + \frac{|\{i \leq m\}| }{|\{i \geq r\}|} \bigg) \cdot \SC(\bo).
\end{align*}
When $n \geq 4$ is even, by the definition of $m$ and $r$, we have that $|\{i \leq m\}| = |\{i \geq r\}| = n/2$, leading to an approximation ratio of at most $2$. When $n \geq 3$ is odd, we have that $|\{i \leq m\}| = (n+1)/2$ and $|\{i \geq r\}| = (n-1)/2$, leading to an upper bound of $1 + |\{i \leq m\}|/|\{i \geq r\}| \leq 2n/(n-1) \leq 3$.
\end{proof}

We now show that the {\sc Median-Right} mechanism is the best possible by showing a matching lower bound of $3$ on the worst-case (over all possible instances) approximation ratio of any deterministic strategyproof mechanism. 

\begin{theorem}\label{thm:max:deterministic:lower}
The approximation ratio of any deterministic strategyproof mechanism is at least $3$. 
\end{theorem}

\begin{proof}
Consider an instance with $n=3$ agents positioned at $0$, $1$, and $2$. Since there are three possible locations for two facilities, one facility must be placed at $0$ or $2$, say $0$. Then, the cost of the agent at position $2$ is at most $2$. Now consider a new instance in which this agent moves to $1$. Due to strategyproofness, one of the facilities must still be placed at $0$ since, otherwise, the agent that moved from $2$ to $1$ would decrease its cost from $2$ to at most $1$. Hence, the social cost of the solution computed by the mechanism is $3$. On the other hand, placing the two facilities at $1$ leads to social cost $1$, and the approximation ratio is at least $3$.
\end{proof}

While no deterministic strategyproof mechanism can achieve an approximation ratio better than $3$ in general, as we have already seen in Theorem~\ref{thm:max:deterministic:upper}, the {\sc Median-Right} mechanism actually has an approximation ratio of at most $2$ when $n$ is an even number. We next show that when the number of agents $n \geq 3$ is odd (which is the worst class of instances for deterministic mechanisms), it is possible to design a randomized strategyproof mechanism with improved approximation ratio of at most $2$. In particular, we consider the following {\sc Uniform} mechanism: 
With probability $1/2$ choose the solution $(\ell,m)$, and with probability $1/2$ choose the solution $(m,r)$. This mechanism is clearly strategyproof-in-expectation as it is defines a constant probability distribution over two deterministic strategyproof mechanisms (the {\sc Median-Left} and the {\sc Median-Right}). 

\begin{theorem}\label{thm:max:randomized:upper}
For any odd $n \geq 3$, the approximation ratio of the {\sc Uniform} mechanism is at most $2$. 
\end{theorem}

\begin{proof}
Since there is an odd number $n \geq 3$ of agents, by the definition of $m$, we have that $|\{i \geq m\}| = |\{i \leq m\}| = (n+1)/2$. 
Hence, we can write the expected social cost of the randomized solution $\bw$ chosen by the mechanism as follows: 
\begin{align*}
\mathbb{E}[\SC(\bw)]
&= \frac12 \bigg( \sum_{i\leq \ell} d(i,m) + \sum_{i \geq m} d(i,\ell) \bigg) + 
\frac12 \bigg( \sum_{i\leq m} d(i,r) + \sum_{i \geq r} d(i,m) \bigg) \\
&= \sum_{i} d(i,m) + \frac12 |\{i \geq m\}| \cdot d(\ell,m) + \frac12 |\{i \leq m\}| \cdot d(m,r) \\
&= \sum_{i} d(i,m) + \frac12 \cdot \frac{n+1}{2} \cdot d(\ell,r).
\end{align*}
For the optimal solution $\bo$, since the position of the median agent is the point that minimizes the total distance from all agents, we have that
\begin{align*}
\SC(\bo) \geq \sum_{i} d(i,m).
\end{align*}
Since there are two facilities to be placed, in $\bo$ one facility must be placed at the position of some agent $o \leq \ell$ or $o \geq r$. 
In the former case, we have that 
$$\forall i \geq r: \cost_i(\bo) \geq d(i,o) = d(i,r) + d(r,m) + d(m,\ell) + d(\ell,o) \geq d(\ell,r).$$
In the latter case, we have that 
$$\forall i \leq \ell: \cost_i(\bo) \geq d(i,o) = d(i,\ell) + d(\ell,m) + d(m,r) + d(r,o) \geq d(\ell,r).$$ 
Since $|\{i \geq r\}| = |\{i \leq \ell\}| = (n-1)/2$ by the definition of $\ell$ and $r$, we have established that
\begin{align*}
\SC(\bo) \geq \frac{n-1}{2} \cdot d(\ell,r).
\end{align*}
Using these two lower bounds on the optimal social cost, we can upper-bound the social cost of $\bw$ as follows:
\begin{align*}
\mathbb{E}[\SC(\bw)]
&\leq \bigg(1 + \frac12 \cdot \frac{n+1}{2} \cdot \frac{2}{n-1} \bigg) \cdot \SC(\bo) = \frac{3n-1}{2n-2} \cdot \SC(\bo). 
\end{align*}
Hence, the approximation ratio is at most $(3n-1)/(2n-2) \leq 2$ for $n \geq 3$. 
\end{proof}

Finally, we show $2$ is the best possible approximation ratio for any randomized strategyproof-in-expectation mechanism. 

\begin{theorem} \label{thm:max:randomized:lower}
The approximation ratio of any randomized strategyproof-in-expectation mechanism is at least $2$.
\end{theorem}

\begin{proof}
We consider the same instance $I$ as in the proof of Theorem~\ref{thm:max:deterministic:lower}. So, there are three agents with locations $0$, $1$, and $2$. Since there are three possible locations for two facilities, there is probability $p \geq 1/2$ that one of the facilities will be placed at $0$ or $2$, say $0$. Then, the expected cost of the agent at position $2$ is equal to $2p + 1-p = p+1$. 

Now consider the instance $J$ in which this agent moves to $1$. If there is probability $q < p$ that a facility is placed at $0$ in $J$, then the agent would have decreased its expected cost from $p+1$ to $q+1$, which contradicts that the mechanism is strategyproof-in-expectation. Hence, one facility must be placed at $0$ with probability at least $p \geq 1/2$ in $J$, which means that the expected social cost is 
$$p \cdot \SC(0,1) + (1-p) \cdot \SC(1,1) = 3p + 1-p = 2p+1 \geq 2.$$
However, the optimal social cost is $\SC(1,1) = 1$, leading to an approximation ratio of at least $2$.
\end{proof}

\section{Deterministic Mechanisms for Multiple Facilities} \label{sec:multiple}
Having completely resolved the case of $k=2$ facilities in the previous two sections, we now consider the general case of $k$ facilities for which we present (asymptotically) tight bounds on the approximation ratio of {\em deterministic} strategyproof mechanisms. 

\subsection{Sum-variant}\label{sec:multiple:sum}
We again start with the sum-variant and first argue about the structure of the optimal solution when there are $k$ facilities to be placed.

\begin{lemma}
For the sum-variant, an optimal solution is to place the facilities at a set of consecutive agents that includes the median agent(s). 
\end{lemma}

\begin{proof}
Let $m$ by a median agent (note that there might be two in case the number of agents is even). The lemma follows directly by the fact that the social cost is the sum of the distances of the agents from {\em all} facility locations and the monotonicity property of the total distance function $f(x) = \sum_i d(i,x)$ that $f(x) \geq f(y) \geq f(m)$ for every $x \leq y \leq m$ or $m \leq y \leq x$.
\end{proof}

We now show our upper bound by considering a generalization of the {\sc Median-Right} mechanism that we used for $k=2$. 
If $k \geq 2$ is even, our mechanism places the facilities at the (leftmost) median agent $m$, at the $k/2-1$ agents at the left of $m$, and at the $k/2$ agents at the right of $m$ (which might include the second median agent in case of an even overall number of agents).
If $k \geq 3$ is odd, the mechanism places the facilities at the (leftmost) median agent $m$, at the $(k-1)/2$ agents at the left of $m$, and at the $(k-1)/2$ agents at the right of $m$.
We will refer to this mechanism as {\sc Median-Ball} (given that it places the facilities around the median agent within a radius of about $k/2$ in each direction). 

Since the mechanism bases its decision only on the ordering of the agents on the line, it is clearly strategyproof for the same reason that {\sc Median-Right} is strategyproof when $k=2$, so in the following we focus on bounding its approximation ratio.

\begin{theorem} \label{thm:sum:multiple:upper}
For the sum-variant, the approximation ratio of the {\sc Median-Ball} mechanism is at most $2$. 
\end{theorem}

\begin{proof}
We present the proof for an odd number $k \geq 3$ of facilities; the proof is similar for even $k$.
Let $\bw = (x_{(k-1)/2}, \ldots, x_1, m, y_1, \ldots, y_{(k-1)/2})$ be the solution computed by the mechanism.
To compute the social cost of $\bw$, we first consider the agents that are not part of the solution. 
Let $S_<$ and $S_>$ be the sets of agents that are to the left of agent $x_{(k-1)/2}$ and to the right of agent $y_{(k-1)/2}$, respectively. 
Also, let $X$ be the indicator variable that is $1$ if $n$ is even and $0$ otherwise.
By definition, we have that $|S_<| = |S_>|-1$ if $X=1$, and 
$|S_<| = |S_>|$ otherwise. 
In any case, since $|S_<| \leq |S_>|$, 
we can match every agent $i \in S_<$ to an agent $\mu(i) \in S_>$ and observe that, for any $w \in \bw$,
$$d(i,w) + d(\mu(i),w) = d(i,m) + d(\mu(i),m).$$
Clearly, if the number of agents is even, there will be an agent $R \in S_>$ that is left unmatched;\footnote{Note that, if $k$ is even, there might be an agent in $S_<$ that is left unmatched instead of an agent in $S_>$.} for this agent $R$, if it exists, we use the fact that $d(x_\ell,R) = d(x_\ell,y_\ell) + d(y_\ell,R)$.
Given this, we have
\begin{align*}
   \sum_{i \not\in \bw} \cost_i(\bw) 
   &= \sum_{i \in S_<} \bigg( \cost_i(\bw) + \cost_{\mu(i)}(\bw) \bigg) + X \cdot \cost_R(\bw) \\
   &= \sum_{i \in S_<} \sum_{w \in \bw} \bigg( d(i,m) + d(\mu(i),m) \bigg) + X\cdot \sum_{w \in \bw} d(R,w) \\
   &= k \cdot \sum_{i \not\in \bw \cup \{R\}} d(i,m) + X \cdot \bigg( \sum_{\ell=1}^{(k-1)/2} d(x_\ell,y_\ell) + 2 \sum_{\ell=1}^{(k-1)/2} d(R,y_\ell) + d(R,m) \bigg).
\end{align*}
Next, we consider the agents that are part of the solution $\bw$ and the distances between them. 
Consider any two agents $x, y \in \bw$ between which there are $t$ different agents.
For each such agent $i \in (x,y)$, we need to take into account the distance of $x$ to $i$, the distance of $i$ to $x$, the distance of $y$ to $i$, and the distance of $i$ to $y$. All together, these distances are exactly
$$2 \left( d(x,i) + d(i,y) \right) = 2\cdot d(x,y).$$ 
Accounting for the agents $x$ and $y$ as well, we have that the contribution of the distances of all agents in $[x,y]$ to the social cost is 
$$(2t+2) \cdot d(x,y).$$
We can now use this observation for all pairs of agents $(x_\ell, y_\ell)$ for $\ell \in [(k-1)/2]$ (note that by doing this we will have calculated the distances of all agents in $\bw$ from all agents in $\bw$, including $m$).  
Since there are $2\ell-1$ agents between $x_\ell$ and $y_\ell$, the distance $d(x_\ell,y_\ell)$ has a coefficient of $4\ell$ 
in the social cost.\footnote{If $k$ is even, for any $(x_\ell, y_\ell)$ for $\ell \in [k/2]$ with $x_1=m$, there are $2\ell-2$ agents between $x_\ell$ and $y_\ell$, leading to a coefficient of $4\ell-2$ for the distance $d(x_\ell,y_\ell)$.} 
Hence, 
\begin{align*}
     \sum_{i \in \bw} \cost_i(\bw)  = \sum_{\ell=1}^{(k-1)/2} 4\ell \cdot d(x_\ell,y_\ell) \leq 2(k-1) \sum_{\ell=1}^{(k-1)/2} d(x_\ell,y_\ell).
\end{align*}
Putting everything together, we have
\begin{align}\label{eq:sum:multiple:mech}
\SC(\bw) 
&\leq k \cdot \sum_{i \not\in \bw \cup \{R\}} d(i,m) + X \cdot \bigg(  2 \sum_{\ell=1}^{(k-1)/2} d(R,y_\ell) + d(R,m) \bigg) \nonumber \\
&\quad + (2k-2+X) \cdot \sum_{\ell=1}^{(k-1)/2} d(x_\ell,y_\ell).
\end{align}

We now focus on bounding the optimal social cost. By Lemma~\ref{sec:multiple:sum}, the optimal solution $\bo$ can be thought of as a shift of $\bw$ towards the left or the right. We will only consider the case where the shift is towards the right; the other case can be handled similarly and is simpler since the agent $R$, if it exists, will have larger cost in the optimal solution, thus leading to a smaller bound on the approximation ratio. 
We again start by considering the agents that are not part of the solution $\bw$. 
As before, consider the same matching $\mu$ of the agents in $S_<$ to the agents in $S_>$. 
Let $o \in \bo$ be some agent that is part of the optimal solution. 
For any agent $i \in S_<$ such that $o \leq \mu(i)$, we have that
$$d(i,o) + d(\mu(i),o) = d(i,m) + d(\mu(i),m).$$
On the other hand, for any agent $i \in S_<$ such that $\mu(i) < o$,
$$d(i,o) = d(i,\mu(i)) + d(\mu(i),o) \geq d(i,m) + d(\mu(i),m).$$
Therefore, 
\begin{align*}
\sum_{i \not\in \bw} \cost_i(\bo) 
&= \sum_{i \in S_<} \bigg( \cost_i(\bo) + \cost_{\mu(i)}(\bo) \bigg) + X \cdot \cost_R(\bo) \\
&= \sum_{i \in S_<} \sum_{o \in \bo} \bigg( d(i,m) + d(\mu(i),m) \bigg) + X \cdot \sum_{o \in \bo} d(R,o) \\
&\geq k \cdot \sum_{i \not\in \bw \cup \{R\}} d(i,m) + X \cdot \bigg( \sum_{\ell=1}^{(k-1)/2} d(R,y_\ell) + d(R,m) \bigg). 
\end{align*}
Next, consider agent $x_\ell$ for $\ell \in [(k-1)/2]$ and let $o \in \bo$. 
If $o \leq y_\ell$, then
$$d(x_\ell,o) + d(o,y_\ell) = d(x_\ell,y_\ell),$$
Otherwise, if $o > y_\ell$, then
$$d(x_\ell,o) = d(x_\ell,y_\ell) + d(y_\ell,o) > d(x_\ell,y_\ell).$$
Hence, we overall have that
\begin{align*}
     \sum_{i \in \bw} \cost_i(\bo) \geq \sum_{\ell=1}^{(k-1)/2} \sum_{o \in \bo} \bigg( d(x_\ell,o) + d(y_\ell,o) \bigg)  \geq k \cdot \sum_{\ell=1}^{(k-1)/2} d(x_\ell,y_\ell).
\end{align*}
Putting everything together, we have
\begin{align*}
\SC(\bo) 
&\geq k \cdot \sum_{i \not\in \bw \cup \{R\}} d(i,m) + X \cdot \bigg( \sum_{\ell=1}^{(k-1)/2} d(R,y_\ell) + d(R,m) \bigg) + k \cdot \sum_{\ell=1}^{(k-1)/2} d(x_\ell,y_\ell) \\
&\geq X \cdot \sum_{\ell=1}^{(k-1)/2} d(R,y_\ell) + k \cdot \sum_{\ell=1}^{(k-1)/2} d(x_\ell,y_\ell).
\end{align*}
It is now not hard to observe that \eqref{eq:sum:multiple:mech} implies
\begin{align*}
\SC(\bw) &\leq \SC(\bo) + X \cdot \sum_{\ell=1}^{(k-1)/2} d(R,y_\ell) +  (k-2+X) \cdot \sum_{\ell=1}^{(k-1)/2} d(x_\ell,y_\ell)  \\
&\leq \SC(\bo) + X \cdot \sum_{\ell=1}^{(k-1)/2} d(R,y_\ell) +  k \cdot \sum_{\ell=1}^{(k-1)/2} d(x_\ell,y_\ell) \\
&\leq 2 \cdot \SC(\bo). 
\end{align*}
This completes the proof. 
\end{proof}

We next provide an asymptotically tight lower bound of $2-1/k$ using a construction that is a generalization of the one in the proof of Theorem~\ref{thm:sum:deterministic:lower} for $k=2$, where instead of having just $3$ agents, we now have $k+1$. 

\begin{theorem} \label{thm:sum:multiple:lower}
For the sum-variant, when there are $k$ facilities to locate, the approximation ratio of any deterministic strategyproof mechanism is at least $2-1/k$.
\end{theorem}

\begin{proof}
Consider an instance with $n=k+1$ agents with one agent at $0$, $k-1$ agents at $1$ (or very close to $1$) and one agent at $2$. Since not all facilities can be placed at $1$, at least one of them has to be placed $0$ or $2$, say $0$. Then, the cost of the agent $i$ that is located at $2$ is at least $k$ (in particular, the cost of $i$ is $2+k-1 = k+1$ if no facility is placed at $2$, and $2+k-2=k$ if a facility is placed at $2$).

Now consider a new instance in which $i$ has moved to $1+\varepsilon$ for some infinitesimal $\varepsilon > 0$. Due to strategyproofness, the mechanism must place one of the facilities at $0$ as well. Otherwise, agent $i$ would have cost $k-\varepsilon$ according to its position in the original instance, and would thus prefer to misreport its position as $1+\varepsilon$ instead of $2$. So, in the new instance, the social cost of any possible solution that is restricted to having a facility at $0$ is approximately $k-1 + k = 2k-1$, while the social cost of the remaining solution is only $k$, leading to an approximation ratio of $2-1/k$. 
\end{proof}

\subsection{Max-variant} \label{sec:multiple:max}
For the max-variant, we will show a tight bound of $k+1$ on the approximation ratio of deterministic strategyproof mechanisms. The upper bound again follows by the {\sc Median-Ball} mechanism; note the upper bound of $2$ on the approximation ratio of {\sc Median-Ball} for the sum-variant immediately implies an upper bound of $2k$ for the max-variant, which however is not the best possible we can show.  

\begin{theorem} \label{thm:max:multiple:upper}
For the max-variant, the approximation ratio of the {\sc Median-Ball} mechanism is at most $k+1$.
\end{theorem}

\begin{proof}
Let $\ell$ and $r$ be the leftmost and rightmost agents in the solution $\bw$ computed by the mechanism. 
By the definition of $\bw$, we have that $||\{i \leq \ell\}| - |\{i \geq r\}|| \leq 1$. 
Since the individual cost of any agent $i$ is the distance to its farthest facility, we have
\begin{align*}
\cost_i(\bw) =
\begin{cases}
d(i,r) & \text{if } i \leq \ell \\
\max\{d(i,\ell), d(i,r)\} & \text{if } i \in \bw \setminus \{\ell,r\} \\
d(i,\ell) & \text{if } i \geq r. 
\end{cases}
\end{align*}
Given this, and using the fact that $d(i,x) \leq d(i,m) + d(m,x)$ for any $x \in \{\ell,r\}$, we can bound the social cost of $\bw$ as
\begin{align*}
\SC(\bw) 
&= \sum_{i \leq \ell} d(i,r) + \sum_{i \in \bw \setminus \{\ell,r\}} \max\{d(i,\ell), d(i,r)\} +  \sum_{i \geq r} d(i,\ell) \\
&\leq \sum_i d(i,m) + |\{i \leq \ell\}| \cdot d(m,r) + (k-2) \cdot \max\{d(\ell,m), d(m,r)\} +  |\{i \geq r\}| \cdot d(\ell,m) \\
&\leq \sum_i d(i,m) + \bigg( \max\big\{ |\{i \leq \ell\}|, |\{i \geq r\}| \big\} + k - 2 \bigg) \cdot  d(\ell,r).
\end{align*}

We now bound the social cost of an optimal solution $\bo$. 
Since the location of the median agent $m$ minimizes the total distance of all agents, if we were allowed to place the facilities at the same location, we would place all $k$ facilities at $m$ to minimize the social cost. Since this is not allowed in our model, the optimal social cost is larger than that, and we obtain
\begin{align*}
    \SC(\bo) \geq \sum_i d(i,m). 
\end{align*}
In addition, since $\bw$ is not optimal (as otherwise the approximation ratio would be $1$), at least one facility must be placed at an agent $o$ that is weakly to the left of $\ell$ or weakly to right of $r$. Let $S$ be the set of agents that are not part of the solution $\bw$ and are on the opposite side of $o$; that is, $S = \{i \geq r\}$ if $o \leq \ell$ and $S = \{r \leq \ell\}$ if $o \geq r$. For each agent $i \in S$, we have that
\begin{align*}
    \cost_i(\bo) \geq d(i,o) \geq d(\ell,r),
\end{align*}
which implies 
\begin{align*}
    \SC(\bo) \geq |S|\cdot d(\ell,r) \geq  \min\big\{ |\{i \leq \ell\}|, |\{i \geq r\}| \big\} \cdot d(\ell,r).
\end{align*}
Putting everything together, we have that
\begin{align*}
    \SC(\bw) \leq \bigg( 1 + \frac{\max\big\{ |\{i \leq \ell\}|, |\{i \geq r\}|\big\} + k-2}{\min\big\{ |\{i \leq \ell\}|, |\{i \geq r\}|\big\}} \bigg) \cdot \SC(\bo).
\end{align*}
Since $\max\big\{ |\{i \leq \ell\}|, |\{i \geq r\}|\big\} \leq \min\big\{ |\{i \leq \ell\}|, |\{i \geq r\}|\big\} + 1$ and $\min\big\{ |\{i \leq \ell\}|, |\{i \geq r\}|\big\} \geq 1$, we obtain an upper bound of $k+1$ on the approximation ratio. 
\end{proof}

We conclude the presentation of our technical results with a matching lower bound of $k+1$ on the approximation ratio of deterministic mechanisms for the max-variant, thus completely resolving this setting. 

\begin{theorem} \label{thm:max:multiple:lower}
For the max-variant, when there are $k$ facilities to locate, the approximation ratio of any deterministic strategyproof mechanism is at least $k+1$.
\end{theorem}

\begin{proof}
Consider an instance with $n=k+1$ agents with one agent at $0$, $k-1$ agents at $1$ (or very close to $1$) and one agent at $2$. Since not all facilities can be placed at $1$, at least one of them has to be placed at $0$ or $2$, say $0$. Then, the cost of the agent $i$ that is located at $2$ is $2$.

Now consider a new instance in which $i$ has moved to $1$. Due to strategyproofness, the mechanism must place one of the facilities at $0$ as well. Otherwise, if all facilities are placed at $1$, agent $i$ would have cost $1$ according to its position in the original instance, and would thus prefer to misreport its position as $1$ instead of $2$. So, in the new instance, the social cost of the solution chosen by the mechanism is $k+1$, while the social cost of the solution that places all facilities at $1$ is just $1$, leading to an approximation ratio of $k+1$. 
\end{proof}

\section{Conclusion and Open Problems} \label{sec:open}
In this work, we showed tight bounds on the best possible approximation ratio of deterministic and randomized strategyproof mechanisms for the heterogeneous two-facility location game where the facilities can be placed at the reported agent locations and the individual cost of an agent is either its distance from both facilities or its distance to the farthest facility. We believe there are many directions for future work. In terms of our results, it would be interesting to close the gap between $2-1/k$ and $2$ for the sum-variant and multiple facilities, consider randomized mechanisms, as well as refine our approximation ratio bounds as functions of the number $n$ of agents. One can also consider the alternative model of homogeneous facilities where the cost of agent is its distance from its closest facility. Furthermore, one can generalize the agent-constrained model in multiple dimensions, for example, by considering agents that might have different preferences over the facilities (such as optional or fractional preferences), and the efficiency of mechanisms is measured by objective functions beyond the social cost (such as the egalitarian cost, or the more general family of $\ell$-centrum objectives). 

\subsection*{Acknowledgments}
Argyrios Deligkas is supported by the UKRI EPSRC grant EP/X039862/1.

\bibliographystyle{plainnat}
\bibliography{references}

\begin{thebibliography}{23}
\providecommand{\natexlab}[1]{#1}
\providecommand{\url}[1]{\texttt{#1}}
\expandafter\ifx\csname urlstyle\endcsname\relax
  \providecommand{\doi}[1]{doi: #1}\else
  \providecommand{\doi}{doi: \begingroup \urlstyle{rm}\Url}\fi

\bibitem[Agrawal et~al.(2022)Agrawal, Balkanski, Gkatzelis, Ou, and Tan]{agrawal2022predictions}
Priyank Agrawal, Eric Balkanski, Vasilis Gkatzelis, Tingting Ou, and Xizhi Tan.
\newblock Learning-augmented mechanism design: Leveraging predictions for facility location.
\newblock In \emph{Proceedings of the 23rd {ACM} Conference on Economics and Computation ({EC})}, pages 497--528, 2022.

\bibitem[Chan et~al.(2021)Chan, Filos{-}Ratsikas, Li, Li, and Wang]{fl-survey}
Hau Chan, Aris Filos{-}Ratsikas, Bo~Li, Minming Li, and Chenhao Wang.
\newblock Mechanism design for facility location problems: {A} survey.
\newblock In \emph{Proceedings of the Thirtieth International Joint Conference on Artificial Intelligence {(IJCAI)}}, pages 4356--4365, 2021.

\bibitem[Chan et~al.(2023)Chan, Gong, Li, Wang, and Zhao]{chan2023ordinal}
Hau Chan, Zifan Gong, Minming Li, Chenhao Wang, and Yingchao Zhao.
\newblock Facility location games with ordinal preferences.
\newblock \emph{Theoretical Computer Science}, 979:\penalty0 114208, 2023.

\bibitem[Chen et~al.(2020)Chen, Fong, Li, Wang, Yuan, and Zhang]{chen2020max}
Zhihuai Chen, Ken C.~K. Fong, Minming Li, Kai Wang, Hongning Yuan, and Yong Zhang.
\newblock Facility location games with optional preference.
\newblock \emph{Theoretical Computer Science}, 847:\penalty0 185--197, 2020.

\bibitem[Cheng et~al.(2013)Cheng, Yua, and Zhang]{cheng2013obnoxious}
Yukun Cheng, Wei Yua, and Guochuan Zhang.
\newblock Strategy-proof approximation mechanisms for an obnoxious facility game on networks.
\newblock \emph{Theoretical Computer Science}, 497:\penalty0 154--163, 2013.

\bibitem[Deligkas et~al.(2023)Deligkas, Filos{-}Ratsikas, and Voudouris]{deligkas2023limited}
Argyrios Deligkas, Aris Filos{-}Ratsikas, and Alexandros~A. Voudouris.
\newblock Heterogeneous facility location with limited resources.
\newblock \emph{Games and Economic Behavior}, 139:\penalty0 200--215, 2023.

\bibitem[Feigenbaum and Sethuraman(2015)]{feigenbaum2015hybrid}
Itai Feigenbaum and Jay Sethuraman.
\newblock Strategyproof mechanisms for one-dimensional hybrid and obnoxious facility location models.
\newblock In \emph{{AAAI} Workshop on Incentive and Trust in E-Communities}, volume {WS-15-08}, 2015.

\bibitem[Feldman et~al.(2016)Feldman, Fiat, and Golomb]{feldman2016voting}
Michal Feldman, Amos Fiat, and Iddan Golomb.
\newblock On voting and facility location.
\newblock In \emph{Proceedings of the 2016 {ACM} Conference on Economics and Computation ({EC})}, pages 269--286, 2016.

\bibitem[Filos{-}Ratsikas et~al.(2024)Filos{-}Ratsikas, Kanellopoulos, Voudouris, and Zhang]{filos2024distributed}
Aris Filos{-}Ratsikas, Panagiotis Kanellopoulos, Alexandros~A. Voudouris, and Rongsen Zhang.
\newblock The distortion of distributed facility location.
\newblock \emph{Artificial Intelligence}, 328:\penalty0 104066, 2024.

\bibitem[Fong et~al.(2018)Fong, Li, Lu, Todo, and Yokoo]{fong2018fractional}
Chi Kit~Ken Fong, Minming Li, Pinyan Lu, Taiki Todo, and Makoto Yokoo.
\newblock Facility location games with fractional preferences.
\newblock In \emph{Proceedings of the 32nd {AAAI} Conference on Artificial Intelligence {(AAAI)}}, pages 1039--1046, 2018.

\bibitem[Fotakis and Tzamos(2014)]{fotakis2014two}
Dimitris Fotakis and Christos Tzamos.
\newblock On the power of deterministic mechanisms for facility location games.
\newblock \emph{{ACM} Transactions on Economics and Computation}, 2\penalty0 (4):\penalty0 15:1--15:37, 2014.

\bibitem[Gai et~al.(2024)Gai, Liang, and Wang]{gai2024mixed}
Ling Gai, Mengpei Liang, and Chenhao Wang.
\newblock Two-facility-location games with mixed types of agents.
\newblock \emph{Applied Mathematics and Computation}, 466:\penalty0 128479, 2024.

\bibitem[Kanellopoulos et~al.(2023)Kanellopoulos, Voudouris, and Zhang]{kanellopoulos2023discrete}
Panagiotis Kanellopoulos, Alexandros~A. Voudouris, and Rongsen Zhang.
\newblock On discrete truthful heterogeneous two-facility location.
\newblock \emph{SIAM Journal on Discrete Mathematics}, 37:\penalty0 779--799, 2023.

\bibitem[Kanellopoulos et~al.(2025)Kanellopoulos, Voudouris, and Zhang]{kanellopoulos2025}
Panagiotis Kanellopoulos, Alexandros~A. Voudouris, and Rongsen Zhang.
\newblock Truthful two-facility location with candidate locations.
\newblock \emph{Theoretical Computer Science}, 1024:\penalty0 114913, 2025.

\bibitem[Li et~al.(2020)Li, Lu, Yao, and Zhang]{li2020constant}
Minming Li, Pinyan Lu, Yuhao Yao, and Jialin Zhang.
\newblock Strategyproof mechanism for two heterogeneous facilities with constant approximation ratio.
\newblock In \emph{Proceedings of the 29th International Joint Conference on Artificial Intelligence ({IJCAI})}, pages 238--245, 2020.

\bibitem[Lotfi and Voudouris(2024)]{lotfi2024max}
Mohammad Lotfi and Alexandros~A. Voudouris.
\newblock On truthful constrained heterogeneous facility location with max-variant cost.
\newblock \emph{Operations Research Letters}, 52:\penalty0 107060, 2024.

\bibitem[Lu et~al.(2010)Lu, Sun, Wang, and Zhu]{Lu2010two-facility}
Pinyan Lu, Xiaorui Sun, Yajun Wang, and Zeyuan~Allen Zhu.
\newblock Asymptotically optimal strategy-proof mechanisms for two-facility games.
\newblock In \emph{Proceedings of the 11th {ACM} Conference on Electronic Commerce ({EC})}, pages 315--324, 2010.

\bibitem[Procaccia and Tennenholtz(2013)]{procaccia09approximate}
Ariel~D. Procaccia and Moshe Tennenholtz.
\newblock Approximate mechanism design without money.
\newblock \emph{{ACM} Transactions on Economics and Computation}, 1\penalty0 (4):\penalty0 18:1--18:26, 2013.

\bibitem[Serafino and Ventre(2016)]{serafino2016}
Paolo Serafino and Carmine Ventre.
\newblock Heterogeneous facility location without money.
\newblock \emph{Theoretical Computer Science}, 636:\penalty0 27--46, 2016.

\bibitem[Tang et~al.(2020)Tang, Wang, Zhang, and Zhao]{Tang2020candidate}
Zhongzheng Tang, Chenhao Wang, Mengqi Zhang, and Yingchao Zhao.
\newblock Mechanism design for facility location games with candidate locations.
\newblock In \emph{Proceedings of the 14th International Conference on Combinatorial Optimization and Applications ({COCOA})}, pages 440--452, 2020.

\bibitem[Xu and Lu(2022)]{Xu2022predictions}
Chenyang Xu and Pinyan Lu.
\newblock Mechanism design with predictions.
\newblock In \emph{Proceedings of the 31st International Joint Conference on Artificial Intelligence ({IJCAI})}, pages 571--577, 2022.

\bibitem[Xu et~al.(2021)Xu, Li, Li, and Duan]{Xu2021minimum}
Xinping Xu, Bo~Li, Minming Li, and Lingjie Duan.
\newblock Two-facility location games with minimum distance requirement.
\newblock \emph{Journal of Artificial Intelligence Research}, 70:\penalty0 719--756, 2021.

\bibitem[Zhao et~al.(2023)Zhao, Liu, Nong, and Fang]{Zhao2023constrained}
Qi~Zhao, Wenjing Liu, Qingqin Nong, and Qizhi Fang.
\newblock Constrained heterogeneous facility location games with max-variant cost.
\newblock \emph{Journal on Combinatorial Optimization}, 45\penalty0 (3):\penalty0 90, 2023.

\end{thebibliography}

\end{document}